\documentclass[12pt]{article}
\usepackage[shortlabels]{enumitem}
\usepackage{amsmath} 
\usepackage{amssymb} 
\usepackage{mathrsfs} 
\usepackage{mathtools} 
\usepackage{amsthm}
\usepackage{thmtools}

\usepackage{appendix}
\usepackage{fullpage}
\usepackage{caption}
\usepackage{subcaption}
\usepackage{graphicx}
\usepackage[table,xcdraw]{xcolor}
\usepackage{wrapfig}
\usepackage{floatrow}
\usepackage{multirow}
\usepackage{makecell}

\definecolor{ForestGreen}{rgb}{0.0333,0.4451,0.0333}
\definecolor{DarkRed}{rgb}{0.65,0,0}
\definecolor{Red}{rgb}{1,0,0}
\usepackage[linktocpage=true,
pagebackref=true,colorlinks,
linkcolor=DarkRed,citecolor=ForestGreen,
bookmarks,bookmarksopen,bookmarksnumbered]{hyperref}

\usepackage[ruled,vlined,linesnumbered]{algorithm2e}
\definecolor{forestgreen}{rgb}{0.13, 0.55, 0.13}

\SetCommentSty{mycommfont}
\SetKwInOut{Input}{Input}
\SetKwInOut{Output}{Output}
\DontPrintSemicolon

\usepackage{cleveref}
\usepackage{thm-restate}
\usepackage[normalem]{ulem} 
\usepackage{nicefrac}

\newtheorem{theorem}{Theorem}[section]
\newtheorem{lemma}[theorem]{Lemma}
\newtheorem{claim}[theorem]{Claim}
\newtheorem{remark}[theorem]{Remark}
\newtheorem{corollary}[theorem]{Corollary}

\newtheorem{definition}[theorem]{Definition}

\newcommand*\mcb[1]{\boldsymbol{\mathcal{#1}}}
\newcommand*\bb[1]{\boldsymbol{#1}}

\DeclarePairedDelimiter\abs{\lvert}{\rvert}
\DeclarePairedDelimiter\norm{\lVert}{\rVert}%

\makeatletter
\let\oldabs\abs
\def\abs{\@ifstar{\oldabs}{\oldabs*}}
\let\oldnorm\norm
\def\norm{\@ifstar{\oldnorm}{\oldnorm*}}
\makeatother

\DeclareMathOperator{\poly}{poly}
\newcommand{\E}{\mathbb{E}}
\newcommand{\OPT}{\text{OPT}}

\newtheorem{assumption}[theorem]{Assumption}

\DeclareMathOperator{\cost}{cost}
\DeclareMathOperator{\dist}{dist}
\DeclareMathOperator{\polylog}{polylog}
\newcommand{\tw}{\mathrm{tw}}
\DeclareMathOperator{\cc}{reach}

\newif\ifcomments
\commentstrue

\ifcomments
\usepackage[colorinlistoftodos,prependcaption,textsize=tiny]{todonotes}

\definecolor{dnotecol}{rgb}{0.20, 0.50, 0.80}

\else 

\fi 

\date{}
\title{The Telephone $k$-Multicast Problem}

\author{Daniel Hathcock
\small Carnegie Mellon University, 
\small USA
\small {\tt dhathcoc@alumni.cmu.edu}
\and 
Guy Kortsarz
\small Rutgers University, Camden, USA 
\small {\tt guyk@camden.rutgers.edu}
\and 
R. Ravi
\small Carnegie Mellon University,
\small USA
\small {\tt ravi@andrew.cmu.edu}
}

\begin{document}

\maketitle

\begin{abstract}
We consider minimum time multicasting problems in directed and undirected graphs: given a root node and a subset of $t$ terminal nodes, multicasting seeks to find the minimum number of rounds within which all terminals can be informed with a message originating at the root. In each round, the telephone model we study allows the information to move via a matching from the informed nodes to the uninformed nodes. 

Since minimum time multicasting in digraphs is poorly understood compared to the undirected variant, we study an intermediate problem in undirected graphs that specifies a target $k < t$, and requires that only $k$ of the terminals be informed in the minimum number of rounds. For this problem, we improve the implications of the previous results and obtain a multiplicative approximation factor of $\tilde{O}(t^{1/3})$. For the directed version, we obtain an {\em additive} $\tilde{O}(k^{1/2})$ approximation algorithm (with a polylogarithmic multiplicative factor). Our algorithms are based on reductions to the related problems of finding $k$-trees of minimum poise (sum of maximum degree and diameter) and applying a combination of greedy network decomposition techniques and set covering under partition matroid constraints. 

We also study the problem of bounded degree Directed Steiner Tree, for which we obtain improved polylogarithmic approximations for the special case of bounded treewidth graphs. This extends prior work on the Group Steiner Tree problem. 
\end{abstract}

\section{Introduction}

We study an information spreading problem that captures applications in distributed computing~\cite{g45} and keeping distributed copies of databases synchronized~\cite{demers1987epidemic}. 
A given graph models a synchronous network of processors that exchange information in rounds. 
There are several models that describe how information may be exchanged between processors in the graph. In this work, we focus on the classic \emph{Telephone Model}~\cite{Hedetniemi1988GossipSurvey}: during a round, each vertex that knows the message can send the message to at most one of its neighbors. 

In the {\bf Minimum Time Telephone Multicast} (MTM) problem, we are given a network, modeled by a directed or undirected graph $G(V,E)$, a root vertex $r$ that knows a message, and a set $S$ of terminals. The message must be transmitted from $r$ to $S$ under the telephone model. In every round, there is a set of vertices $K \subseteq V$ that know the message (initially $K = \{r\}$), and communication in a given round is described by a matching $\{(k_1, v_1), \dots, (k_\ell, v_\ell)\}$ between some pairs of vertices $k_i \in K$ and $v_i \not \in K$ for which $k_iv_i \in E$. In the directed setting, the edge $k_iv_i$ must be directed from $k_i$ to $v_i$. Following this round, all matched vertices $\{v_i\}$ are added to $K$. When $S=V$ this problem is called 
{\bf The Minimum Time Broadcast} (MTB) problem.

The best known approximation ratio for the MTM problem on an undirected graph is due to Elkin and Kortsarz~\cite{g21}. They achieve $O(\log t/\log\log t)$, where $t=|S|$. In \cite{EK0}, it is shown that unless $P=NP$, the MTB problem does not admit an approximation $3-\epsilon$ for any constant $\epsilon$. For directed graphs, the Minimum Time Broadcast problem admits an $O(\log n)$ approximation \cite{EK0} in an $n$-node graph. The same paper shows that unless $P=Quasi(P)$ the problem admits no better than $\Omega(\sqrt{\log n})$ approximation.
 
However, for the directed case the multicast problem {\em seems} harder to approximate. The best known approximation ratio for this problem is an additive guarantee of $O(\sqrt{t})$ (with a polylogarithmic multiplicative factor) \cite{g24}. This leaves a wide gap between the current best approximation algorithms for undirected versus directed multicast problems. In this work, we make progress toward closing that gap by studying an intermediate problem, the \textbf{Minimum Time Telephone $k$-Multicast} problem ($k$-MTM), defined below.
\begin{center} \fbox{\begin{minipage}{0.965\textwidth} 
{\bf Input:} A directed or undirected graph $G(V,E)$ with root $r$, a collection
of terminals $S\subseteq V$  and a number $k\leq|S|$.  

{\bf Required:} 
Send the message originating at $r$ to {\em any} $k$ terminals of $S$ in the telephone model
in a minimum number of rounds.
\end{minipage}}\end{center}
In terms of approximability, the undirected $k$-MTM problem lies between the undirected and directed MTM problems. Specifically, in \cite{g18} it is shown\footnote{\cite{g18} deals with the degree-bounded versions of these problems, but their proof works as well for poise (maximum degree + diameter) problems. See below for the connection between poise and $k$-MTM.} that a $\rho$-approximation for directed MTM implies an $O(\polylog k) \cdot \rho$-approximation for undirected $k$-MTM, while it is immediate that any approximation for undirected $k$-MTM gives the same factor approximation for undirected MTM.

On the other hand, the directed version of the $k$-MTM problem generalizes all the aforementioned problems.

\paragraph{Applications} 
Broadcast and multicast problems find numerous applications in distributed settings. For example, in the Network Aggregation problem, each user sends its data to a chosen central vertex $r$. This is equivalent to broadcasting in the local model for distributed computation (see \cite{g54}).
Broadcasting is also crucial in sensor networks \cite{g23}.
Another application is to ensure that the maximum information delay in vector
clock problems is minimized \cite{g75,g20}.

An application of multicasting is to keep the information between copies of replicated databases consistent, by broadcasting the changed copy to the others \cite{g71,nr14,g72}.
If we are given a large set of $t$ terminals of which we only want to keep replicated copies in some $k$ of them, finding the best $k$ to minimize the maximum synchronization time between these terminals corresponds to the $k$-MTM problem.

\paragraph{Minimum Poise Trees} Any telephone multicast schedule defines a tree rooted at $r$, 
spanning all terminals. The parent of a vertex $u\neq r$ is defined to be the 
unique vertex that sends the message to $u$.
Let $T^*$ be the tree defined by the optimal schedule. 
The height of $T^*$ (the largest distance in $T^*$ from the root) 
is denoted by $D^*$. The largest out-degree\footnote{For simplicity, we say degree instead of out-degree for the rest of the chapter when discussing directed graphs.} in $T^*$ is denoted by $B^*$. 
The {\em poise}\footnote{Poise is often defined as the sum of maximum degree and diameter, rather than height, in settings where there is no prescribed root. Since our setting has a root, we choose to use height. This cannot change any result by more than a factor 2.} of $T^*$ is defined as $p^*=B^*+D^*$
\cite{g56}. Denote by $\OPT$ the number of rounds used by the optimal schedule.
Since at every round, each informed vertex can send the message 
to at most one neighbor, $\OPT \geq B^*$ and $\OPT \geq D^*$. Hence, in general we have
$\OPT \geq \nicefrac{p^*}{2}$. A partial converse is shown in \cite{g56}. A     
$\rho$ approximation for the \textbf{Minimum Poise Steiner Tree} implies an $O(\log t)\cdot \rho/\log\log t$ approximation for the MTM problem.

An identical approach to \cite{g56} yields the equivalence (up to logarithmic factors in $k$) of approximating the $k$-MTM problem and approximating the following \textbf{Minimum Poise Steiner $k$-Tree} problem:

\begin{center} 
\fbox{\begin{minipage}{0.965\textwidth} 
{\bf Input:} A directed or undirected graph $G(V,E)$ with root $r$, a collection
$S\subseteq V$ of terminals, and a number $k$.

{\bf Required:} 
A $k$-tree rooted at $r$, namely a tree $T'(V,E)$ containing 
paths from $r$ to $k$ of the terminals, with minimum poise.
\end{minipage}}\end{center}

\begin{definition}
A $\tilde O(f(k))$-additive approximation for a minimization problem with optimal cost $\OPT$ is an algorithm that returns a solution whose cost is $\tilde O(\OPT) + \tilde O(f(k))$.
\end{definition}

We include the proof of the following lemma in the preliminaries for completeness. 
\begin{lemma}[\cite{g56}]\label{lem:poise-to-mtm}
    A $\tilde O(f(k))$-additive approximation for the Minimum Poise Steiner $k$-Tree problem implies an $\tilde O(f(k))$-additive approximation for the $k$-MTM problem.   
\end{lemma}

In \cite{g18}, they show that the approximability of minimum degree Steiner $k$-Tree reduces to minimum degree Group Steiner Tree (which is a special case of minimum degree Directed Steiner Tree). Their reduction immediately extends to the minimum poise versions of these problems. Hence, the approximability of the undirected Minimum Poise Steiner $k$-Tree problem lies between the undirected Minimum Poise Steiner Tree problem and the directed Minimum Poise Steiner Tree problem (up to $\log k$ factors). This implies the aforementioned analogous statement about the relationship between undirected $k$-MTM and the undirected/directed MTM problems. 

We focus on approximating these poise problems.

\paragraph{Bounded Treewidth Graphs} We further consider the setting in which it is assumed that the input graph has bounded treewidth. Trees and series-parallel graphs fall into this family, and the setting has been studied as an important special case for numerous graph algorithm problems (e.g.~\cite{BHM11,GTW13,CHLN05,CDLV17,DKL24}, to name a few).

We consider the \textbf{Bounded Degree Directed Steiner Tree} problem with edge weights in the bounded treewidth setting:

\begin{center} 
\fbox{\begin{minipage}{0.965\textwidth} 
{\bf Input:} A directed graph $G(V, E)$ with root $r$ and edge weights $\{w_e\}_{e \in E}$, a collection of terminals $S \subseteq V$, and an out-degree bound $d_v$ for every vertex $v \in V$.

{\bf Required:} 
A minimum-weight directed Steiner tree rooted at $r$, that is, a tree $T'(V,E)$ containing a path from $r$ to every terminal that satisfies the out-degree bound for every vertex.
\end{minipage}}\end{center}

The algorithms and analysis for the poise problems considered in this chapter separately bound the degree and height of the returned tree. The degree bound makes up the bulk of the analysis, while the height bound is a simple observation. Interestingly, for the above Bounded Degree Directed Steiner Tree problem on bounded treewidth graphs, our analysis addresses the degree bounds, but it appears difficult to also bound the height of the tree in order to give a statement about poise in this setting as well. 

\subsection{Our results}

We give an $\tilde O(\sqrt{k})$-additive approximation for the directed versions of the $k$-MTM (and Minimum Poise Steiner $k$-Tree) problems. 
\begin{theorem}
\label{t1}
Minimum Poise Steiner $k$-tree 
problem on directed graphs
admits a polynomial time  
$\tilde O(k^{1/2})$-additive 
approximation. 
This implies the same approximation for the Minimum Time Telephone $k$-multicast problem.
\end{theorem}
The second part of the statement follows from \Cref{lem:poise-to-mtm}. 

In \cite{rohit1}, a multiplicative $O(\sqrt{k})$-approximation is given for the directed Min-Max Degree $k$-Tree problem, which asks to find a tree spanning $k$ terminals while minimizing the maximum degree. Their algorithm iteratively finds trees containing $\sqrt{k} \cdot B^*$ terminals and uses flows to connect them to the root. Our directed result is more general than that of~\cite{rohit1} 
in that it can handle both degree bounds and height bounds.
Moreover, our approximation for the degree is stronger, since we get an additive approximation of $O(\sqrt{k})$. Therefore, it may be better than the approximation of \cite{rohit1} in the case where $B^*$ is large. 
Our approximation ratio for the height is constant. 

Our result is also more general than the $\tilde O(\sqrt{t})$-additive approximation for the directed MTM of \cite{g24}, as it handles the $k$-tree version of the problem, and recovers the same $\tilde O(\sqrt{t})$-additive approximation in the case $k = t$ (up to logarithmic factors). In \cite{g24}, the so-called \emph{multiple set-cover} problem is used, a variant of set cover, while our result uses maximum coverage subject to a matroid constraint. 

For undirected graphs, we give an 
$\tilde O(t^{1/3})$ approximation, which is a better ratio in the worst case if $k$ is close to $t$. This represents progress toward closing the gap between the approximability of undirected and directed MTM, since in \cite{g18} it is shown that the undirected $k$-MTM problem lies between undirected and directed MTM in terms of approximability. 
\begin{theorem}
\label{t2}
The Minimum Poise Steiner $k$-tree 
problem on undirected graphs admits a polynomial time $\tilde O(t^{1/3})$ approximation and, therefore, the Minimum Time Telephone $k$-Multicast problem
admits the same approximation.
\end{theorem}
 
The $O(\sqrt{k})$ additive ratio can be as bad as the
$\Omega(\sqrt{t})$ multiplicative ratio, if $B^*$ is constant and
$k=\Omega(t)$.
Therefore, in the worst case, an $\tilde O(t^{1/3})$ approximation is  
a better ratio.
In addition, if $B^*=o(t^{1/6})$ and $k=\Omega(t)$, 
the multiplicative ratio gives a better additive ratio.

For the bounded degree Directed Steiner Tree problem on bounded treewidth graphs, we give a polylogarithmic approximation algorithm which returns a Steiner tree violating the degree constraints by only a polylogarithmic factor. 

\begin{restatable}{theorem}{dbdst}
\label{thm:dbdst}
The bounded degree Directed Steiner Tree problem on bounded treewidth graphs admits a polynomial time $O(\log^2 n)$ approximation algorithm with an $O(\log^2 n)$ factor violation in the degree bounds.
\end{restatable}

This extends the work of Dinitz, Kortsarz, and Li~\cite{DKL24} who gave an analogous algorithm for the bounded-degree Group Steiner Tree problem. 

\subsection{Technical Overview}
For the directed case, our techniques are based on \cite{g24}.
However, our problem is harder since it is not clear which $k$ terminals 
to choose. An important difference is that we use an approximation algorithm to maximize coverage (a submodular function) under matroid constraints \cite{check33}.  
The multiplicative approximation for the undirected case builds on this and requires
several graph decomposition techniques to be carefully combined.

For both results, we denote the maximum degree as $B^*$ and height as $D^*$ of an optimal minimum poise tree $T^*$. It can be assumed that $D^*$ and $B^*$ are known by trying all possibilities, as there are only polynomially many. Moreover, since $D^*$ is known, all vertices of distance greater than $D^*$ from the root may be removed.

\paragraph{Directed Min-Poise Steiner $k$-Tree}
In order to get an $O(\sqrt{k})$ additive approximation for the directed min-poise Steiner $k$-tree problem, we employ a greedy strategy. We iteratively find a collection of vertex-disjoint trees, each covering (that is, containing) exactly $\sqrt{k}$ terminals and of height at most $D^*$, until no more can be found. We call these \textit{good trees}. 

In the case where at least $\sqrt{k}$ many good trees are found, taking any $\sqrt{k}$ of the good trees along with the shortest paths from the root $r$ to the roots of each of these trees yields an additive $O(\sqrt{k})$-approximation. This yields a subgraph (not necessarily a tree, since the shortest paths may not be disjoint from the good trees) with maximum out-degree at most $2\sqrt{k}$, and radius (maximum distance from $r$) at most $2\cdot D^*$. Moreover, the subgraph contains $k$ terminals. Now the non-disjointness may be overcome by returning a shortest path tree spanning this subgraph. This gives the desired approximation. 

In the other case that fewer than $\sqrt{k}$ good trees are found, we may still connect them to the root via shortest paths. This gives a subgraph of low poise, but does not yet cover $k$ terminals. If $k_1 < k$ terminals are covered, we must determine how to cover $k - k_1$ additional terminals without inducing high degree or height. 

This is the main technical contribution of the directed result: we can recast the covering of $k - k_1$ additional terminals as a set cover instance, and the desired poise guarantees can be obtained by imposing a partition matroid constraint on the sets in the instance. Then, an algorithm for approximating submodular function maximization subject to a matroid constraint~\cite{check33} is applied. To the authors' knowledge, partition matroid constrained maximum coverage has not previously been used for multicasting problems. The matroid constraint guarantees that the additional degree incurred on already-chosen vertices by covering these terminals is low. The degree incurred on newly added vertices is shown to be bounded because none of these could be chosen to be the root of a good tree. 

\paragraph{Partition Matroid Max Coverage Procedure}
Suppose we are given a partition of the graph into $A \cup C = V$ with $r \in A$, such that all of $A$ is reachable with low poise and contains $k_1$ terminals. We want to cover at least $k - k_1$ terminals in $C$ with low poise, and we know that there exists a tree $T^*$ rooted at $r$ which does so.

Say that a node $c \in C$ \textit{covers} all the terminals in $C$ that it can reach within distance $D^*$. In this way, we define a set cover instance over the ground set of terminals in $C$ in which each set is identified by an edge $(a, c)$ between a node $a \in A$ and a node $c \in C$. The set corresponding to $(a, c)$ contains all terminals covered by $c$. Defining the sets this way allows us to enforce degree constraints on the nodes in $A$, since the sets can be partitioned by their member in $A$. That is, we form a partition with the parts $X(a) = \{(a, c) : c \in C, ac \in E\}$ for each $a \in A$. We now impose the constraint that at most $B^*$ sets may be chosen from any part $X(a)$, reflecting the desired degree constraint. A \emph{partition matroid} captures choosing at most a certain number of elements from each part of a partitioned set. Hence we have described a set cover instance with a partition matroid constraint and a coverage requirement of $k-k_1$. 

The problem of selecting sets to maximize the number of terminals covered subject to the matroid constraint is a special case of submodular function maximization subject to a matroid constraint. Moreover, $T^*$ provides a certificate that there exists a collection of sets satisfying the matroid constraint and covering at least $k - k_1$ terminals in $C$.
Hence, we may apply the $(1 - \frac{1}{e})$-approximation for this problem~\cite{check33} (the simple greedy strategy, which gives a $\frac{1}{2}$-approximation~\cite{Fisher1978} would also suffice here) to find a collection of sets satisfying the matroid constraint and covering at least $(1 - \frac{1}{e}) \cdot (k - k_1)$ terminals in $C$. 

Given the choice of sets $(a, c)$ by the algorithm, we identify a set of edges that may be added to extend our subgraph to cover these terminals. These newly covered terminals are then removed, and the process repeated. In each round, we can cover a constant fraction of the desired number of terminals, so we need only $O(\log k)$ rounds. Moreover, any given round induces additional degree of only $B^*$ on nodes in $A$. The degree induced on nodes in $C$ depends on the size of the parts $X(a)$, and this can be bounded in our applications (e.g., by $\sqrt{k}$ in the directed setting described above, since the greedy strategy ensures that all nodes in $c$ can reach at most $\sqrt{k}$ terminals within distance $D^*$). 
Finally, the distance from the root of any node added is $O(D^*)$, so in total the poise of the subgraph remains low. In the end, we again output a shortest path tree spanning this subgraph. 

\paragraph{Improvement in Undirected Graphs}
In the undirected setting, the result can be improved by taking advantage of the fact that if a good (low-poise) tree covering many terminals is found, then we need only cover \emph{any} node in that tree in order to cover all of those terminals with low poise (as opposed to the directed case where we would have to cover the \emph{root} of that tree). Essentially, we may contract the tree and treat the contracted node as containing many terminals. 

Specifically, we will maintain a set $R$ of nodes that we have covered so far with low poise (by contracting, we can think of this simply as the root $r$). We first group the terminals in the remaining graph $C = V \setminus R$ as before by greedily finding disjoint trees of low poise, now each containing $t^{1/3}$ terminals, called \emph{small trees}. Note that some terminals may not lie in any small tree. If the algorithm finds fewer than $t^{1/3}$ small trees, then the same matroid-constrained covering procedure from above can be applied to immediately get an \textit{additive} $O(t^{1/3})$-approximation. 

On the other hand, if there are many small trees, we show that progress can be made by either covering or discarding a large number of terminals at once. If we are able to aggregate $t^{1/3}$ small trees into a single tree within a distance $D^*$, we have covered $t^{2/3}$ terminals and hence made sufficient progress in coverage: we can repeat this at most $t^{1/3}$ times to finish, inducing at most $t^{1/3}$ degree at the root to reach these trees. However, we may have the additional complexity of the optimal tree containing terminals that are not in one of these small trees we computed in $C$. 
We handle this case by using the matroid-constrained coverage procedure to extract as many terminals as any optimal solution might cover from the small trees while staying within the degree and height bounds, and then discarding all the terminals from \emph{all} of the unused small trees. Since the number of small trees (each with $t^{1/3}$ terminals) is $\Omega(t^{1/3})$, this allows us to bound the number of such discarding iterations by $O(t^{1/3})$. In summary, we employ $O(t^{1/3})$ iterations of either covering or discarding $t^{2/3}$ terminals in the algorithm leading to the claimed $O(t^{1/3})$ multiplicative guarantee. Over the course of these iterations, the total degree accumulated by any node will be at most $\tilde O(t^{1/3}) \cdot B^*$ (Note this guarantee is now multiplicative, since a node can gain $\tilde O(B^*)$ degree in each of the $t^{1/3}$ covering iterations). 


Finally, we remark that the improved guarantee in this setting is in terms of $t$, the total number of terminals, rather than $k$. This is because our algorithm relies on removing a large number of terminals from the entire set of $t$ terminals, without necessarily covering all of them. 

\paragraph{Bounded Treewidth Graphs}
In the setting of graphs with bounded treewidth, we use the techniques of Dinitz, Kortsarz, and Li~\cite{DKL24} to give an improved approximation for the minimum-weight Directed Steiner Tree problem with degree bounds (as opposed to a poise bound). They define the so-called Tree Label Problem in which the nodes of a tree must be labeled in a way which satisfies consistency, covering, and cost constraints. Using LP rounding, they give a randomized algorithm that finds a consistent labeling that approximately satisfies each covering and cost constraint with good probability. 

We reduce the bounded degree Directed Steiner Tree Problem to the Tree Label Problem, and then use their Tree Label algorithm as a black box to get the desired result. For a graph with bounded treewidth, the tree decomposition is used as the input to the Tree Label instance. The label sets for a node (bag) in the tree decomposition represent the connectivity relations between graph vertices in the corresponding bag. Then, the covering constraints capture the connection of the root to each terminal, while the cost constraints are used to bound both the degrees and the total cost of the solution. 

Our reduction, while very similar to that of \cite{DKL24} for bounded degree Group Steiner Tree, differs in our choice of label sets. Since the problem of \cite{DKL24} is undirected, their labels include a partition of the vertices in a bag that represents the connected components. Our problem is on a directed graph, so we have to more carefully represent the reachability relationship between the vertices in the bag.

\section{Preliminaries}

\paragraph{Steiner Poise} Let $\dist(u,v)$ denote the number of edges in the shortest path 
from $u$ to $v$ in $G$.
We denote  by $G[U]$ the graph induced by $U$, and by 
$\dist_{G[U]}(u,v)$ the distance from $u$ to 
$v$ in the graph $G[U]$.
Recall that we denote the minimum poise tree by $T^*$, its maximum degree by $B^*$, and its height by $D^*$. 

\begin{assumption}\label{ass:dist}
Removing vertices of distance more than $D^*$ from the root $r$ in $G$ does not change the optimal solution. Hence, we will assume for the rest of the chapter that $G$ only contains vertices of distance at most $D^*$ from $r$. 
\end{assumption}

\begin{remark}
For the rest of the chapter, we assume that quantities 
such as $\sqrt{k}$ are integral. 
Making the algorithm precise requires using $\lceil \sqrt{k}\rceil$.
However, the changes are minimal and elementary.
\end{remark}

For simplicity, we assume that 
every terminal has in-degree $1$ 
and out-degree $0$, by attaching new terminal vertices 
to every terminal (this only increases the poise by at most 2). 
For undirected graphs, we assume that terminals 
have degree $1$. Therefore, removing terminals 
can't turn a connected graph into a disconnected 
graph.

Now we give the proof of \Cref{lem:poise-to-mtm}, which follows identically to the analogous proof of \cite{g56} for the versions of these problems in which all terminals must be covered. 
\begin{proof}[Proof of \Cref{lem:poise-to-mtm}]
The algorithm for $k$-MTM is as follows: using the $\tilde O(f(k))$-additive approximation for Minimum Poise Steiner $k$-Tree, find a Steiner $k$-Tree $T$. Suppose $T$ has poise $p$. Then using the result of \cite{g56}, find a broadcast schedule on $T$ that informs all of the nodes in $T$ of a message originating at $r$ using no more than $O(\frac{\log \abs{T}}{\log \log \abs{T}}) \cdot p$ rounds. This clearly gives a $k$-multicast schedule on the original graph taking the same number of rounds. 
 
Consider an optimal $k$-multicast schedule, taking $\OPT$ rounds. It suffices to show that $p \leq \tilde O(\OPT) + \tilde O(f(k))$. Using the optimal schedule, define the tree $T^\OPT$, rooted at $r$, by connecting each informed vertex $u \neq r$ to the unique vertex from whom it received the message. 
Say $T^\OPT$ has height $D^\OPT$ and maximum degree $B^\OPT$, then its poise is $p^\OPT :=B^\OPT+D^\OPT$. Since at every round, each informed vertex can send the message 
to at most one neighbor, $\OPT \geq B^\OPT$ and $\OPT \geq D^\OPT$. Hence, in general we have
$\OPT \geq \nicefrac{p^\OPT}{2}$. 

But now if $p^*$ is the poise of the minimum poise Steiner $k$-tree, then $p^* \leq p^\OPT$, and hence we know by the approximation algorithm that 
\[
    p \leq \tilde O(p^*) + \tilde O(f(k)) \leq \tilde O(\OPT) + \tilde O(f(k))
\]
as desired. 
\end{proof}

\paragraph{Max Coverage with a Matroid Constraint} The input for the  Set Cover  problem is a
universe $\mcb{U}$ and a collection 
$\mcb{S}$ of sets $S_i \subseteq \mcb{U}$. We say that a set $S_i$  
{\em covers} all the elements that belong to this set. 
The goal is to find a sub-collection of sets 
$\mcb{S}'\subseteq \mcb{S}$ of minimum size that covers all elements, 
namely, $\bigcup_{S_i\in \mcb{S}'} S_i=\mcb{U}$. 
The Maximum Coverage problem under a matroid constraint has the input of 
Set Cover, and in addition, a matroid $\mcb{M}$ defined over the sets $\mcb{S}$. The goal is to select an independent set  $\mcb{I}$ in the Matroid so that  
$|\bigcup_{S_i\in \mcb{I}} S_i|$ is maximum.
A partition matroid instance divides  
$\mcb{S}$ into pairwise disjoint collections of sets 
 $\mcb{S}_i$, whose union is all 
 of $\mcb{S}$. For every collection $\mcb{S}_i$,
there is a bound $p_i$ on the number of sets that 
can be selected from $\mcb{S}_i$. A collection of sets containing at most $p_i$ sets from each $\mcb{S}_i$ is precisely an independent set in the partition matroid.
The goal is to find an independent set in the partition matroid 
that covers the largest number of elements.
This problem 
is a special case of maximizing a submodular function under matroid 
constraints. The greedy algorithm achieves a $\nicefrac{1}{2}$-approximation for this problem~\cite{Fisher1978}, and is sufficient for our purposes. It is also known that the problem 
admits a polynomial time $1-\nicefrac{1}{e}$-approximation
\cite{check33}, which may be used for improved constants. The procedure of 
\cite{check33} is one of the main tools in our algorithm.
We call this procedure the \texttt{Matroid} procedure.

\paragraph{Treewidth} A \emph{tree decomposition} of a graph $G = (V, E)$ is a tree $\bb T = (B, \bb E)$ where each node $b \in B$ is called a \emph{bag}, and is associated with a subset $X_b$ of $V$, satisfying:
\begin{enumerate}
    \item Every vertex is in at least one bag $X_b$.
    \item The set of bags containing a vertex $v \in V$ forms a connected subtree of $\bb T$.
    \item For every edge $e \in E$, there is a bag containing both endpoints of $e$. 
\end{enumerate}
The width of a tree decomposition is its largest bag size, minus 1, and the \emph{treewidth} of $G$ is the width of its minimum width tree decomposition.

\section{The Partition Matroid Cover Algorithm}

In the next two sections, our algorithms for both the directed and undirected cases define a disjoint partition of the graph vertices into $A \cup C = V$.
The root $r$ always belongs to $A$, and we will ensure that all of $A$ can be covered by a low poise tree rooted at $r$.
In this section, we discuss
how to cover sufficiently many terminals from $C$ with low poise by connecting them to the root through $A$.
We do this by defining an instance of the Max Coverage problem under a partition matroid constraint\footnote{
Note that the parameter $k$ represents the {\em remaining} number of terminals we need to cover.
Given a partition $A,C$ we will assume that all terminals in $A$ have been spanned, and thus we need to cover 
$k$ terminals in $C$. 
That is, if $A$ has $k_1$ terminals for some $k_1<k$, we will set $k \gets k-k_1$. Note that we are guaranteed that 
$C\cap T^*$ contains at least $k-k_1$ terminals supplying a feasible solution.
}.

\begin{definition}
\label{d2}
Define a Max Coverage instance as follows.
\begin{itemize}
\item The items are $S\cap C$ (the terminals in $C$).
\item The sets (also called \emph{pairs}) are 
$\mcb{S} = \{(a, c)\mid a\in A,c\in C, \text{ and } ac \in E\}$ where $(a, c)$ \emph{covers} a terminal $t \in S \cap C$ if $\dist_{G[C]}(c,t)\leq D^*$. 
\end{itemize}
\end{definition}
\noindent The partition matroid is defined as follows.
\begin{definition}
$\mcb{S}$ is partitioned into collections 
\[
    X(a)=\{(a, c)\mid c\in C  \text{ and }  ac\in E\}
\]
 for every $a \in A$. The bound on the number of sets to be chosen 
from $X(a)$ is $B^*$.
\end{definition}
By definition, the partition is disjoint and therefore, we have a valid partition matroid. Recall that $r\in A$. See \Cref{alg:pmcover} for a description of the Procedure \texttt{PMCover}.

\begin{algorithm}[h]	
    \caption{\texttt{PMCover}}
    \label{alg:pmcover}
    
    \Input{Graph $G(V, E)$ with terminals $S$ and $V$ partitioned into $A \cup C$, and a number $k$.}
    \Output{A collection of pairs of the form $(a, c)$ with $a \in A$ and $c \in C$.}
    \BlankLine

    $\mcb{E}'\gets \emptyset$, $S'\gets S\cap C$.\;
    \While{$k > 0$}{
        Define the partition matroid Set Coverage instance from $A,C,S'$ as above with sets $\mcb{S}'$ and apply Procedure \texttt{Matroid} of \cite{check33} to find an independent set of approximately maximum coverage. Let $\mcb{I}$ be the independent set it returns.\;
        
        $\mcb{E}'\gets \mcb{E}'\cup \mcb{I}$.\;
        
        Decrease $k$ by the number of terminals covered by $\mcb{I}$.\;

        Remove the terminals covered by $\mcb{I}$ from $S'$.\;
    }
    
    \Return $\mcb{E}'$.
\end{algorithm}

\subsubsection*{Analysis}
We will show that 
for every $a\in A$, $\abs{X(a)\cap \mcb{E}'} \leq O(\log k)\cdot B^*$.
This will be used to argue that if $(a, c)\in \mcb{E}'$, we later may make 
$a$ the parent of $c$ in the tree we build without incurring high degree. 

\begin{definition}
\label{d1}
Define a mapping from terminals in $T^*\cap C$  
to $\mcb{S}'$ as follows. 
For a terminal $t$, let  $a=a_t$ be the 
vertex $a\in A $ that is an ancestor 
of $t$ in $T^*$ and among them 
$\dist_{T^*}(a,t)$ is minimum. 
This vertex is well defined since $r\in A$ is the root of $T^*$.
Let $c=c_t$ be the child of $a$ in $T^*$ that is an ancestor of $t$.
Define $f(t)=(a, c)$.
\end{definition}

\begin{claim}
\label{c11}
There exists an independent set $\mcb{I}^*$ in the partition matroid that covers at least $k$ terminals in $C\cap S$.
\end{claim}
\begin{proof}
We show that every terminal
in $t\in T^*\cap C$ is covered by some set.
Let $a=a_t$ and let $c = c_t$. Since $a$ has minimum 
distance to $t$ from all vertices in $A$, the path from 
$c$ to $t$ belongs to $G[C]$.
The number of edges in the path between $c$ and $t$ is at most $D^*-1$.
This implies that 
the set $(a, c)$ covers $t$.
Create a set $\mcb{I}^* = \{f(t) \mid t \in T^* \cap S \cap C\}$.
We note that $f(t)=f(t')=(a, c)$ may hold for two different terminals, but $\mcb{I}^*$ includes every such pair $(a, c)$ once  
(namely, $\mcb{I}^*$ is a set and not a multiset). 
For any $a \in A$, the number of different pairs of the form $(a,c_1),(a, c_2),\ldots$ in $\mcb{I}^*$ can't be more than $B^*$, because every such pair increases $a$'s degree in $T^*$ by $1$. Thus, $\mcb{I}^*$ is independent in the partition matroid. 
Since all terminals in $T^*\cap C$ are covered, $k$ terminals are covered.
\end{proof}

\begin{claim}
\label{c5}
Procedure \texttt{PMCover} returns a collection of pairs $\mcb{E}'$ so that for every $a\in A$, $X(a)\cap \mcb{E}'=O(\log k)\cdot B^*$ 
and $\mcb{E}'$ covers $k$ terminals.
Thus if in some tree, vertex $a \in A$ is made the parent of all $c$ 
for which $(a,c)\in \mcb{E}'$, 
the degree of $a$ will be bounded by 
$O(\log k)\cdot B^*$.
\end{claim}
\begin{proof}
Since Procedure \texttt{Matroid} returns an independent set in the partition matroid, at every iteration we have
$|X(a)\cap \mcb{I}|\leq B^*$.
Claim \ref{c11} and the guarantee of 
Procedure \texttt{Matroid} by \cite{check33} imply that  
$(1-\nicefrac{1}{e})k$ terminals are covered.
Let $k_{\text{or}} \leq k$ be the original number 
of terminals to be covered 
and $k_{\text{new}}$ the number of terminals 
to be covered in a given iteration.
Then in the next iteration, 
\[
    k_{\text{new}}\gets k_{\text{new}}-(1-\nicefrac{1}{e})k_{\text{new}}=\frac{k_{\text{new}}}{e}.
\]
Therefore, after 
$i$ iterations, $k_{\text{or}}/e^i$ terminals 
remain to be covered.
Hence, the number of iterations is $O(\log k)$.
The claim follows.
\end{proof}

\section{Approximating the poise for directed graphs}
Our algorithm maintains a set $A$ (initialized with the root $r$) containing the terminals covered with low poise so far, and $C = V \setminus A$. 
Consider a set $C$ and the graph $G[C]$ induced 
by $C$.
\begin{definition}
\label{d3}
Denote by $T(c)$ the coverage tree of $c$ in $G[C]$ formed by taking a shortest path from $c$ to every terminal within distance $D^*$. A vertex $c\in C$ is called 
$\rho$-good (with respect to $C$) if there are at least 
$\rho$ terminals in $T(c)$. A $\rho$-good tree 
is a tree rooted at some $c$ (not necessarily $T(c)$) with {\em exactly} $\rho$ terminals and height at most $D^*$.
\end{definition}
By assumption, the out-degree of terminals 
is $0$. Therefore all terminals are leaves.
Since we may discard non-terminal leaves, a 
$\rho$-good tree contains exactly $\rho$ leaf terminals.

\begin{definition}
A  set  $C$ of  vertices, is a $\rho$-packing if there is no $\rho$-good vertex in $C$.
\end{definition}
\begin{definition}
Let $\{T_i\}$ be a collection of 
vertex disjoint trees and  let $A$ be the set 
of vertices in $\bigcup_i T_i$. Let $C=V-A$. 
Then $A,C$ is a $\rho$-additive partition if: 
\begin{enumerate}
\item The trees $T_i$ are $\rho$-good with respect to $V$, and are all vertex-disjoint. 
\item There are at most $\rho$ trees $T_i$.
\item $C$ is a $\rho$-packing.
\end{enumerate}
\end{definition}

Let $q_i$ be the root of $T_i$.
Intuitively, since there are at most 
$\rho$ trees $T_i$, we can add a shortest path $P_i$ 
from the root $r$ to each $q_i$, giving a tree rooted at $r$ 
with low poise covering terminals in $A$. In addition, since 
$C$ is a $\rho$-packing, at least $k$ (meaning the number of \textit{remaining} terminals to cover after covering those in $A$) of $C$'s terminals 
can be covered with some collection of low poise trees. In particular, since every $c \in C$ is not $\rho$-good, all of the coverage trees $T(c)$ have max degree at most $\rho$. 

The algorithm attempts to find a $\rho$-additive partition. It greedily finds $\rho$-good trees, and removes them until the set $C$ that remains is a $\rho$-packing. Then the procedure \texttt{PMCover} can be used to connect the low poise trees covering $A$ and $C$. However, there may be too many $\rho$-good trees in $A$ for $(A, C)$ to be a $\rho$-additive partition. In this case, it simply connects the root to any $\rho$ of the trees $T_i$. By choosing $\rho = \sqrt{k}$, this ensures enough terminals are covered. See \Cref{alg:directed} for a precise description of the Procedure \texttt{Directed}. 

\begin{algorithm}[h]	
    \caption{\texttt{Directed}}
    \label{alg:directed}
    
    \Input{Graph $G(V, E)$ with terminals $S$, and a number $k$.}
    \Output{A Steiner $k$-tree of $G$.}
    \BlankLine

    Set $\rho = \sqrt{k}$.\;
    \tcc{Greedy Packing}
    Let $A = \{r\}$, and $C = V - \{r\}$.\;
    \While{$C$ is not a $\rho$-packing}{
        Find a $\rho$-good tree $T$ in $G[C]$.\;
        Remove the vertices of $T$ from $C$ and add them to $A$.\;
    }
    Let $\{T_i\}$ denote the set of $\rho$-good trees found.\;

    \tcc{Many Trees}
    \If{the number of $\rho$-good trees found is at least $\rho$,}{
        Choose any $\rho$ of the trees $\{T_i\}$ in $A$, and form the subgraph $H \subseteq G$ by including the root $r$, the chosen trees, and a shortest path from $r$ to the root $q_i$ of each chosen tree $T_i$.\;
        \Return a shortest path tree of $H$ rooted at $r$.\;
    }
    \tcc{Few Trees}
    \Else{
        The number of $\rho$-good trees found is at most $\rho$, so $(A, C)$ is a $\rho$-additive partition. Apply the Procedure \texttt{Complete} on $(A, C)$, and \Return the resulting tree.\;
    }
\end{algorithm}

In the case that a $\rho$-additive partition $(A, C)$ is found, we use the Procedure \texttt{Complete}, described in \Cref{alg:complete}. See \Cref{fig:directed} for a depiction of the algorithm at this step. 

\begin{algorithm}[h!]	
    \caption{\texttt{Complete}}
    \label{alg:complete}
    
    \Input{Graph $G(V, E)$ with terminals $S$, a $\rho$-additive partition $(A, C)$, and a number $k$. }
    \Output{A Steiner $k$-tree of $G$.}
    \BlankLine

    Apply the procedure \texttt{PMCover} with partition $(A, C)$ to get $\mcb{E}'$.\;
    Let $\mcb E = \{ac: (a, c) \in \mcb E'\}$, the set of arcs corresponding to sets in $\mcb{E}'$\;
    Form the graph $H_C$ on vertex set $C \cup \{r'\}$, where $r'$ is a new node. For each $c \in C$ appearing in some $ac \in \mcb E$, include in $H_C$ the arc $(r', c)$ and the coverage tree $T(c)$. Take a shortest path tree on $H_C$ rooted at $r'$, and let $T_C$ be all of the edges from this tree in $G[C]$. \label{succ:line3}\;
    Form the subgraph $H \subseteq G$ by including the root $r$, each $\rho$-good tree $T_i$ from $A$ and a shortest path from $r$ to its root $q_i$, the edges from $\mcb E$, and the edges from $T_C$.\label{succ:line4}\;
    \Return a shortest path tree of $H$ rooted at $r$. \;
\end{algorithm}

\begin{figure}[h]
    \centering
    \includegraphics[scale=0.65]{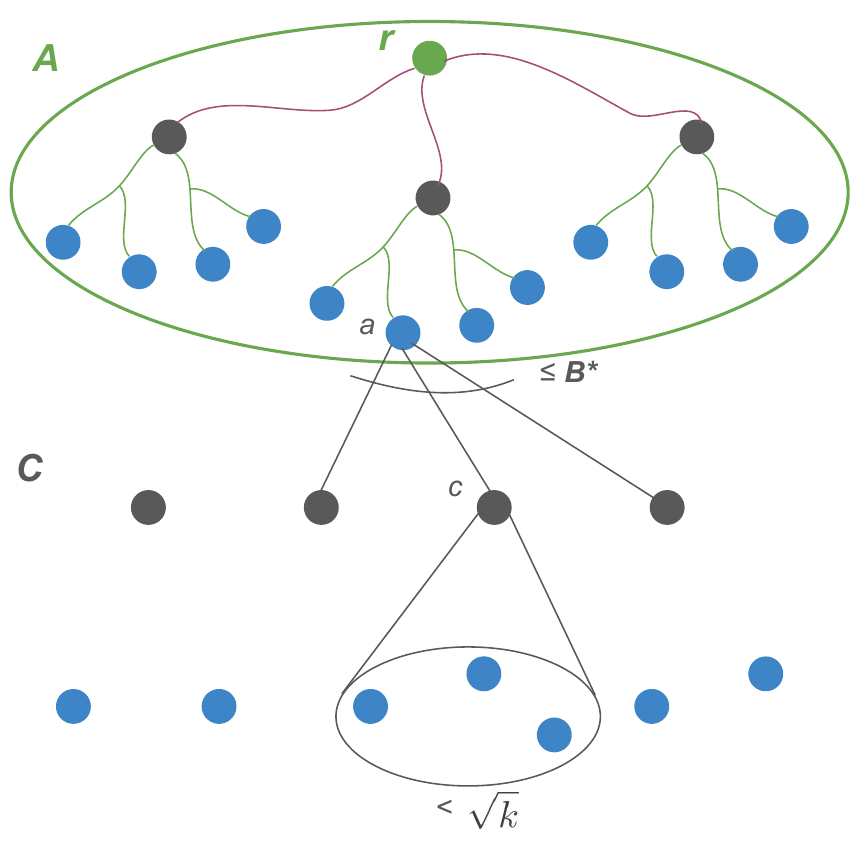}
    \caption{A depiction of the algorithm in the case that a $\rho$-additive partition is found. The set $A$ includes the root $r$ and all $\sqrt{k}$-good trees found, while $C$ contains the remaining vertices. Terminals are depicted in blue. Short paths from $r$ to the roots of the good trees are added (in red). Since $C$ is a $\sqrt{k}$-packing, each vertex $c \in C$ can reach less than $\sqrt{k}$ terminals within distance $D^*$. Hence, we can run the \texttt{PMCover} procedure, with each iteration enforcing a degree constraint of $B^*$ on each node in $A$, as shown.}
    \label{fig:directed}
\end{figure}

\subsubsection*{Analysis}
For a directed tree, $T$, let $\deg_{T}(v)$ be the (out-)degree of the vertex in $T$. Now say that we run step \emph{Greedy Packing} of \texttt{Directed} with $\rho=\sqrt{k}$.

\begin{claim}
\label{c21}
If Procedure \texttt{Directed} finds at least $\rho$ $\rho$-good trees,
then step \textit{Many Trees} of Procedure \texttt{Directed}
returns a tree with at least $k$ terminals, maximum degree $O(\sqrt{k})$,
and height $O(D^*)$
\end{claim}
\begin{proof}
Since each tree $T_i$ is $\sqrt{k}$-good, it contains $\sqrt{k}$ terminals. Hence, the graph $H$ contains at least $k$ terminals, each of which can be reached by a path from the root. So the returned shortest path tree of $H$ has at least $k$ terminals, as desired.

To bound the degrees in the returned tree, we just bound the degrees in $H$. The good trees $T_i$ are disjoint, and each have maximum degree at most $\sqrt{k}$. Moreover, there are $\sqrt{k}$ of them, so there are only $\sqrt{k}$ shortest paths to their roots. Therefore, the degree contributed to any node $v \in H$ is at most $\sqrt{k}$ from the $T_i$, and at most $1$ for each shortest path, for a total of $\deg_H(v) \leq 2\sqrt{k}$. Finally, each tree $T_i$ in $H$ has height at most $D^*$, while each shortest path from the root to some $q_i$ has length at most $D^*$ (by \Cref{ass:dist}), so the returned shortest path tree has height at most $2\cdot D^*$. 
\end{proof}

\begin{claim}
\label{claim:success}
If Procedure \texttt{Directed} finds less than $\rho$ $\rho$-good trees,
then Procedure \texttt{Complete}
finds a tree rooted at $r$ 
with maximum degree $O(\log k)\cdot B^*+O(\sqrt{k})$, and height $O(D^*)$ that contains at least $k$ terminals of $C\cap S$.
\end{claim}
\begin{proof}
First, observe that $H$ contains all terminals in $A$, as well as those terminals in $C$ covered by procedure \texttt{PMCover}. In particular, by \Cref{c5}, $H$ contains at least $k$ terminals, so the returned shortest path tree does as well.

Now we bound the degrees of nodes in the returned tree. The $T_i$ making up $A$ are disjoint $\sqrt{k}$-good trees each having maximum degree at most $\sqrt{k}$. And there are less than $\sqrt{k}$ of them, so we add at most $\sqrt{k}$ shortest paths to their roots $q_i$. Hence, for each node $v \in A$, the contribution to the degree $\deg_H(v)$ is at most $\sqrt{k}$ from the $T_i$, at most 1 for each shortest path, plus the contribution from $\mcb{E}$. By \Cref{c5}, the edges of $\mcb{E}$ increase the degree of vertices in $A$ by $O(\log k)\cdot B^*$, so in total $\deg_H(v) \leq O(\log k) \cdot B^* + 2\sqrt{k}$ for each $v \in A$. 

All other vertices in $H$ lie in $C$, and so their degree comes only from the $\sqrt{k}$ shortest paths (contributing at most 1 each), and the edges from $T_C$. Every coverage tree $T(c)$ has depth at most $D^*$ by definition, and hence any proper subtree of $T_c$ must also have depth at most $D^*$. In particular, for any vertex $c \in C$, we must have $\deg_{T_C}(c) \leq \sqrt{k}$, since otherwise the subtree of $T_C$ rooted at $c$ has more than $\sqrt{k}$ leaves, which can all be assumed to be terminals. But this means that $c$ has more than $\rho=\sqrt{k}$ terminals in $C$ of distance at most $D^*$, contradicting that $c$ is not $\rho$-good. Hence, $\deg_{H}(c) \leq 2\sqrt{k}$ for every $c \in C$.

Finally, the height of the output tree is at most $3\cdot D^*+1$, because we get height $D^*$, from the trees $T_i$, height $D^*$ from the shortest paths, height $D^*$ from $T_C$, and an additional edge from $\mcb{E}$. 
\end{proof}

Therefore, in either case we return a tree with at least $k$ terminals 
with maximum degree 
$O(\log k)\cdot B^*+O(\sqrt{k})$
and height $O(D^*)$. This implies 
Theorem \ref{t1}. 

The following corollary is useful as it applies 
in case that the \emph{Greedy Packing} step of Procedure \texttt{Directed} finds a $\rho$-additive partition (i.e., step \textit{Few Trees} is executed)
with some $\rho$ that may be smaller than 
$\sqrt{k}$.

\begin{corollary}
\label{cor1}
If Procedure \texttt{Directed} finds a $\rho$-additive partition 
$A,C$, then there exists a polynomial time $\rho$-additive 
approximation for the corresponding min poise $k$-tree problem.
\end{corollary}

\section{The undirected case}
In this section, we provide our $\tilde O(t^{1/3})$-approximation algorithm for the Minimum Time Telephone $k$-multicast problem on undirected graphs with $t$ terminals, proving Theorem~\ref{t2}.

\paragraph{Preliminaries}
Recall that we assume that all terminals 
have degree $1$. 
Hence, after rooting $T^*$ at $r$, the set of (non-root) leaves 
in $T^*$ and the set of terminals in $T^*$ is the same set.
Also recall that the height of the tree $T^*$ rooted at  $r$ is $D^*$.

\paragraph{Algorithm outline} The idea in the undirected case is that if a low-poise tree covering many terminals is found, then we need only cover \emph{any} node in that tree in order to cover all of those terminals with low poise (as opposed to the directed case where we would have to cover the \emph{root} of that tree). Essentially, we may contract the tree and treat the contracted node as containing many terminals. 

Specifically, we will maintain a set $R$ of nodes we have covered with low poise via edges $E(R)$ (by contracting $E(R)$, we can think of this simply as the root $r$). We first partition the remaining graph $C = V \setminus R$ as before by greedily finding small trees. 

\begin{definition}
We say that a tree is {\em small} if it contains 
exactly  $t^{1/3}$ terminals. 
We say that a tree is {\em large}  
if it contains exactly  $t^{2/3}$ 
terminals.
\end{definition}

If this procedure succeeds in finding a $t^{1/3}$-additive partition, then we are done by \Cref{cor1}. On the other hand, if we fail, we contract these small trees and show how to cover a sufficiently large number of them by either finding a single large tree reaching $t^{1/3}$ of these small trees, or by applying the procedure \texttt{PMCover}. In either case, we may then remove \emph{all} of the terminals from these small trees, contract the newly covered nodes into $R$, and iterate the entire process to cover the remaining terminals. In each iteration, we show the total number of terminals discarded is large, so there cannot be too many iterations, and hence not too much additional degree is incurred in $E(R)$. 

We first give a simple algorithm, Procedure \texttt{Small}, that finds trees $\{T_i\}$ each with exactly $t^{1/3}$ terminals.

\begin{algorithm}[h!]	
    \caption{\texttt{Small}}
    \label{alg:small}
    
    \Input{Graph $G(V, E)$ with $t$ terminals $S$, and a number $k$. }
    \Output{A collection of subtrees $\{T_i\}$, each with exactly $t^{1/3}$ terminals, or a Steiner $k$-tree.}
    \BlankLine

    Apply step \textit{Greedy Packing} from Procedure \texttt{Directed} on $G$ with $\rho=t^{1/3}$. Denote the resulting trees as $\{T_i\}$.\;

    If the procedure succeeds in finding a $t^{1/3}$-additive partition, apply Procedure \texttt{Complete} on $A,C$, and \Return the resulting tree.\;

    Else, \Return $\{T_i\}$\;
\end{algorithm}

In case 
that step \textit{Greedy Packing} from Procedure \texttt{Directed} finds a $t^{1/3}$-additive partition 
$A,C$, we are guaranteed a $t^{1/3}$-additive ratio 
from \Cref{cor1}.
Hence, from now on we 
assume that step \textit{Greedy Packing} from Procedure \texttt{Directed} gives more than  $t^{1/3}$ small 
trees 
$T_i$. 

We will proceed to contract each of these small trees into super-terminals. 
The trees $T_i$ that we compute are built by 
step \textit{Greedy Packing} from Procedure \texttt{Directed} with $\rho=t^{1/3}$. Hence, they have exactly
$t^{1/3}$ terminals. For each $T_i$, we contract its edges to create a single super-terminal 
$q_i$. Denote by $S(T_i)$ the terminals contained in $T_i$ (i.e., those corresponding to $q_i$). 
As mentioned in the outline, we have the possibility that the terminals of an optimal tree may only intersect with a few of these super-terminals. 
We capture this  in the following definitions.
\begin{definition}
We say that $q_i$ is a {\em true terminal}
if $S(T_i)\cap T^*\neq \emptyset$.
\end{definition}
\begin{definition}
Denote by $k'$ the number of terminals 
in $(\bigcup_i T_i)\cap T^*$.
Let 
$\mu=\lceil k'/t^{1/3} \rceil$. 
\end{definition}
From the definitions, we can see that $T^*$ overlaps with at least $\mu$ true terminals.

In the graph where the small trees have been contracted to super-terminals, we will attempt to find a $t^{1/3}$-packing of these super-terminals. 
For this, we generalize the definition of a $t^{1/3}$-packing in the  
set $C$ with respect to the super-terminals. 
\begin{definition}
We say that $c\in C$ is a $t^{1/3}$-good vertex 
with respect to the super-terminals 
$\{q_i\}$
if there are at least 
$t^{1/3}$ terminals  $q_i$ of distance at most $D^*$ 
from $c$, in $G[C]$. If there are no $t^{1/3}$-good vertices 
in $C$, $C$ is called a $t^{1/3}$-packing with respect to
$\{q_i\}$.
If $C$ is a $t^{1/3}$-packing, then $R,C$ is called 
a $t^{1/3}$-additive partition with respect to $\{q_i\}$. 
\end{definition}

We can now describe the details of the rest of the undirected algorithm. Specifically, if Procedure \texttt{Small} fails to find a $t^{1/3}$-additive partition, then there are two possibilities. Either $C = V \setminus R$ is a $t^{1/3}$-packing with respect to the $q_i$, or otherwise there is a $t^{1/3}$-good vertex in $C$. 

If 
$C$ is a $t^{1/3}$-packing  
we apply Procedure 
\texttt{PMCover} on $R,C$ with terminals 
$\{q_i\}$ since $R,C$ is a $t^{1/3}$-additive partition.
The goal is 
covering $\mu$ super-terminals. We know that $T^*$ covers at least 
$\mu$ true terminals $q_i$, so these can be reached with 
height $D^*$ 
and maximum degree $B^*$.  
Therefore, our Procedure 
\texttt{PMCover} covers at least $\mu$ super-terminals. Note that the number of original terminals we actually cover is  
$\mu\cdot t^{1/3} \geq k'$. This follows because each $q_i$ represents a tree
$T_i$ that contains $t^{1/3}$ terminals. We now add to $E(R)$ the edges used to cover those super-terminals along with the edges from the corresponding trees $T_i$ of the super-terminals covered. Then we discard 
all the terminals 
of $\bigcup_i T_i$. Since the number of 
$T_i$ is at least $t^{1/3}$, the total number of discarded 
terminals is $t^{2/3}$.

The other case is that $C$ is not a $t^{1/3}$-packing 
with respect to $\{q_i\}$. Let $v\in C$ be a 
$t^{1/3}$-good  vertex and let $Q_v$ be the corresponding tree.
Note that $Q_v$ is a large tree since it spans 
$t^{1/3}$ of the $q_i$, each representing $t^{1/3}$ terminals. We connect 
$R$ to $Q_v$ via a shortest path $P$ from $R$ to $Q_v$, add the edges of $P$ and those from $Q_v$ (including the edges from the contracted trees $T_i$) to $E(R)$, and contract 
$R\cup P\cup Q_v$ into $R$. Then we discard the terminals of 
$Q_v$.
Since $Q_v$ is a large tree, the number of terminals discarded is 
$t^{2/3}$.  

In summary, in both cases $t^{2/3}$ terminals 
are discarded. Therefore the number of iterations in our 
algorithm is at most $t^{1/3}$.

The degree of vertices in $R$ induced by $E(R)$ increases  
by $O(\log k)\cdot B^*$ every time 
\texttt{PMCover} is applied.
Alternatively, a large tree $Q_v$ is created 
and we only need a path $P$ from $r$ to $Q_v$. 
This increases the degree of some vertices in $R$ by exactly $2$.
This gives a total degree of $2\cdot t^{1/3}$ because of the bound on the number of iterations.

\subsubsection*{The main procedure}
Here we describe the precise algorithm for the undirected problem, Procedure \texttt{Undirected} in \Cref{alg:undirected}.

\begin{algorithm}[h!]	
    \caption{\texttt{Undirected}}
    \label{alg:undirected}
    
    \Input{Graph $G(V, E)$ with $t$ terminals $S$, and a number $k$.}
    \Output{A Steiner $k$-tree}
    \BlankLine
    $R\gets \{r\}$, $E(R) \gets \varnothing$, $S'\gets S$.\;
    \While{$k > 0$}{
        Apply Procedure \texttt{Small} with $\rho=t^{1/3}$ on $C=V\setminus R$. If it succeeds, return the resulting tree.\;

        If \texttt{Small} fails, contract the terminals from each $T_i$ in the resulting packing into a corresponding super-terminal $q_i$.\;

        \If{$C=V\setminus R$ is not a $t^{1/3}$-packing with respect to $\{q_i\}$}{
            Find a large tree $Q_v$ inside $G[C]$.\;
            Compute a shortest path $P$ from $r$ to $Q_v$.\;
            $R \gets R \cup P \cup Q_v$.\;
            $E(R) \gets E(R) \cup E(P) \cup E(Q_v)$.\;
            Remove from $S'$ all the terminals of $Q_v$ and update $k$.\;
        }
        \Else(\tcp*[h]{$C = V \setminus R$ is a $t^{1/3}$-packing.}){
            Apply Procedure \texttt{PMCover} with $A=R$ and $C=V-R$ with the goal of covering super-terminals. \;
            Let $\mcb{E}$ be the edges corresponding to the returned sets $(a, c) \in \mcb{E}'$, and $T_C$ the shortest path tree on the corresponding trees $T(c)$ (as in line \ref{succ:line3} of \texttt{Complete}). Write $Q = \mcb{E} \cup T_C$.\;
            $R \gets R \cup Q$.\;
            $E(R) \gets E(R) \cup E(Q)$.\;
            Remove from $S'$ all the terminals $\bigcup_i T_i$ and update $k$.\;
        }
    }
    \Return the tree induced by $E(R)$.\;
\end{algorithm}

\subsubsection*{Analysis}

\begin{claim}
\label{claim:true-terms}
$T^*$ contains at least 
$\mu$ true terminals.
\end{claim}
\begin{proof}
If the number of true terminals is at most $\mu-1$,
the number of terminals in  $\bigcup_i T_i$ is at most $(\mu-1)\cdot t^{1/3}<k'$ 
and this is a contradiction.
\end{proof}

\begin{claim}
\label{number}
The number of iterations 
in 
Procedure \texttt{Undirected} is at most
$t^{1/3}$.
\end{claim}
\begin{proof}
If a tree $Q_v$ is found, then it is a large tree 
hence it contains at least $t^{2/3}$ terminals. 
These terminals are discarded in the iteration.
Else, the terminals of $S\cap \bigcup_i T_i$ 
are discarded and this,  again, this removes 
$t^{2/3}$  terminals, since we have at least $t^{1/3}$ different small $T_i$'s (because procedure \texttt{Small} failed). 
Since in either case $t^{2/3}$ 
terminals are discarded and the total number 
of terminals is $t$, the number of iterations is
at most $t^{1/3}$.
\end{proof}

\begin{claim}
\label{before}
Let $v$ be a vertex so that $v\not\in R$.
A single iteration of 
Procedure \texttt{Undirected}
increases $v$'s degree in $E(R)$ by at most $2 \cdot t^{1/3} + 2$. Moreover, if $v$'s degree increases, $v$ is contracted 
into $r$ in that iteration.
\end{claim}
\begin{proof}
The degree of a vertex 
increases only if it belongs to a 
large tree $Q_v$ (or its path $P$ from $r$), or it belongs to the subgraph $Q$ 
computed by Procedure \texttt{PMCover}.
In the first case, the degree increases by at most $t^{1/3}$ from any one of the $T_i$'s in $Q_v$, at most $t^{1/3}$ more for the paths from $v$ to these $T_i$'s, and at most 2 more for the path $P$ from $r$ to $Q_v$, for a total of at most $2\cdot t^{1/3} + 2$. In the second case, the degree increases by at most $t^{1/3}$ from $T_C$ and 1 from $\mcb{E}$, by an identical argument to \Cref{claim:success} (the proof of correctness of Procedure \texttt{Complete}). 

In both cases, 
the vertex is immediately contracted into $r$.
\end{proof}

\begin{claim}
\label{after}
At every iteration, the degree of 
vertices in $R$ is increased by at most 
$O(\log k)\cdot B^*$. 
\end{claim}
\begin{proof}
If a large tree $Q_v$ is found, a shortest 
path from $r$ to $Q_v$ is computed.
This increases the degree of any  vertex by at most $2$.
Otherwise, Procedure \texttt{PMCover} is applied. 
This increases the degree of vertices of $R$
by $O(\log k)\cdot B^*$. The claim follows.
\end{proof}

\begin{claim}
The returned tree contains $k$ terminals, 
has maximum degree 
$\tilde O(t^{1/3})\cdot B^*$ and height 
$O(D^*)$
\end{claim}
\begin{proof}
By Claim \ref{before}, 
an iteration of Procedure \texttt{Undirected} increases 
the degree of a vertex $v \not \in R$ by at most $O(t^{1/3})$, and any $v$ whose degree increases immediately joins $R$. 
Now we bound 
the degree added to a vertex in $R$.
By Claim 
\ref{after} at every iteration the degree 
of $v\in R$ can increase by 
$O(\log k)\cdot B^*$. 
By Claim \ref{number}, the number of iterations 
of Procedure \texttt{Undirected} is bounded by 
$t^{1/3}$. Therefore the total degree 
of a vertex is at most 
\[
    O(t^{1/3})+O(\log k)\cdot B^* \cdot t^{1/3}=O(\log k)\cdot t^{1/3} \cdot B^*.
\]
In addition, the diameter of every $Q_v$ or $Q$ found, is
$O(D^*)$.
The distance of $r$ to any $Q$ or $Q_v$ is at most 
$D^*$ as well by \Cref{ass:dist}. This ensures that the height of the tree is $O(D^*)$.

Finally, we argue that $k$ terminals are covered. Fix a particular iteration of the algorithm. If Procedure \texttt{Small} succeeds in finding a $t^{1/3}$-additive partition, then we immediately cover the remaining number of terminals necessary by applying the Procedure \texttt{Complete}. Otherwise, we argue that among the terminals discarded in this iteration, the algorithm covers at least as many as $T^*$ covers. Indeed, if $C$ is not a $t^{1/3}$-packing with respect to $\{q_i\}$, then all terminals discarded are covered. On the other hand, if $C$ is a $t^{1/3}$-packing, then by applying Procedure \texttt{PMCover}, \Cref{claim:true-terms} ensures that at least $\mu$ super-terminals are covered. Hence at least $\mu \cdot t^{1/3} \geq k'$ terminals are covered, which is precisely the number of terminals covered by $T^*$ among those discarded. 
\end{proof}

Using 
\Cref{lem:poise-to-mtm},
we get the following 
corollary that proves Theorem~\ref{t2}.
\begin{corollary}
The Minimum Time Telephone $k$-Multicast 
problem on undirected graphs 
admits a polynomial time,  $\tilde O(t^{1/3})$-approximation algorithm.
\end{corollary}

\section{Bounded Degree Directed Steiner Tree}


In this section we address the related problem of \textbf{Bounded Degree Directed Steiner Tree}. We get as input a directed graph $G(V, E)$ with root $r$ and edge weights $\{w_e\}_{e \in E}$, a collection of terminals $S \subseteq V$, and an (out-)degree bound $d_v$ for every vertex $v \in V$. The goal is to find a minimum-weight directed Steiner tree, that is, an out-arborescence connecting the root to every terminal, that satisfies the (out-)degree bound for every vertex. 

We prove \Cref{thm:dbdst}, restated below, for graphs of bounded treewidth. This extends the analogous result of Dinitz, Kortsarz, and Li~\cite{DKL24} for the Group Steiner Tree problem with degree bounds on graphs of bounded treewidth. 

\dbdst*

Our methodology follows that of \cite{DKL24}: a reduction to the so-called Tree Label Problem.

\subsection{The Tree Label Problem Theorem}

We briefly describe the Tree Label Problem, introduced by Dinitz, Kortsarz, and Li~\cite{DKL24}. An instance is a full binary tree $\bb{T} = (\bb V, \bb E)$ rooted at $\bb r \in \bb V$, and for each $u \in \bb V$ there is a finite label set $L_u$ (all disjoint). Denote $L := \bigcup_{u \in \bb V} L_u$. The output is a labeling $\vec \ell = (\ell_u \in L_u)_{u \in \bb V}$ satisfying: 
\begin{itemize}
    \item \textit{Consistency constraints}: labels must come from some $\Gamma_u \subseteq L_u \times L_v \times L_{v'}$ for each parent and children pair $u$ and $v, v'$. 
    \item \textit{Covering constraints}: there are $k$ subsets $S_1, \dots, S_k$ of labels, and at least one label must be chosen from each subset. 
    \item \textit{Cost constraints}: there are $m \geq 0$ linear costs $c_\ell^i \in [0, 1]$. We must satisfy $\sum_{u \in \bb V} c^i_{\ell_u} \leq 1$ for every $i$. 
\end{itemize}

Let $n = \abs{\bb V}$, $D$ be the height of $\bb T$, and $\Delta = \max_{u \in \bb V} \abs{L_u}$. They give a randomized algorithm for this problem which we will use as a black box in the reduction. 

\begin{theorem}[\cite{DKL24}]\label{thm:tree-label}
Given a feasible (i.e. there exists a valid labeling) tree label instance $(\bb T, \bb r, (L_u)_{u}, (\Gamma_u)_u, (S_t)_{t}, \bb c \in [0, 1]^{m \times L})$, there is a randomized $\poly(n) \cdot \Delta^{O(D)}$-time algorithm that outputs a consistent labeling $\vec \ell$ such that 
\begin{enumerate}[(a)]
    \item For every $t \in [k]$, we have $\Pr[\vec \ell \text{ covers group } S_t] \geq \frac{1}{D}$. 
    \item For every $i \in [m]$, we have $\E[\exp(\ln(1 + \frac{1}{2D}) \cdot \cost^i(\vec \ell))] \leq 1 + \frac{1}{D}$.
\end{enumerate}
\end{theorem}

\subsection{Reduction to the Tree Label Problem}

Consider an instance of the bounded degree Directed Steiner Tree problem in which the given graph $G = (V,E)$ with root node $r$ has bounded treewidth $\tw$. That is, there is a tree decomposition $\bb T = (B, \bb E)$ with bags $b \in B$ of size $\tw + 1$. The vertex set of a bag is denoted $X_b$. Bodlaender showed~\cite{B88} that we may further assume that it has the following properties. 
\begin{itemize}
    \item $\bb T$ is a full binary tree rooted at $\bb r$, with depth $O(\log n)$ where $n = |V|$. 
    \item $\abs{X_b} \leq O(1) \cdot \tw$ and $r \in X_b$ for every $b$.
    \item For every edge $(u, v) \in E$, there is some $b$ with $u, v \in X_b$
    \item For every $v \in V$, the set of bags containing $v$ is connected in $\bb T$. 
\end{itemize}
For each edge $e$, say $b_e$ is the highest bag in $\bb T$ containing $e$, and for a bag $b$, $E_b$ is the set of edges for which $b_e = b$ (all the edges for which $b$ is ``responsible''). We also use $\Lambda(b)$ to denote the set of descendants (including $b$) of $b$ in $\bb T$.

Given such an instance and the associated tree decomposition, we now construct a tree label instance. The idea is that our label sets will describe the connectivity structure that a subgraph (i.e. the output directed Steiner tree) induces on the vertices in each bag, while the consistency constraints ensure that a single subgraph can be recovered from a labeling. The covering and cost constraints are then used to ensure feasibility of this subgraph (each terminal is reachable from the root) and our approximation guarantees, respectively. 

Given a subgraph $H$, we define the ``state information'' of $H$ for each bag: 
\begin{itemize}
    \item $F_b(H)$ is the set of edges $H$ contains from $E_b$.
    \item $\Pi^\downarrow_b(H)$ is a transitive relation on $X_b$ that tells for each $u, v \in X_b$ whether $v$ is reachable from $u$ in the subgraph induced on the descendant bags $(V_H, E_H \cap \bigcup_{b' \in \Lambda(b)} E_{b'})$. 

    \item $\Pi^\uparrow_b(H)$ is a transitive relation on $X_b$ that tells for each $u, v \in X_b$ whether $v$ is reachable from $u$ in the subgraph induced on the non-descendant bags. That is, the subgraph $(V_H, E_H \cap \bigcup_{b' \in (B \setminus \Lambda(b)) \cup \{b\}} E_{b'})$.
\end{itemize}
Now, we define the label set $L_b$ for each $b \in B$ to be all tuples $(F_b, \Pi^\downarrow_b, \Pi^\uparrow_b)$ that correspond to some valid $H$. Since we think of $H$ as a valid output to the DST problem, we can assume that $F_b$ is a branching (a forest of arborescences) for each $b$. Furthermore, if $b = \bb r$ the root of $\bb T$, we must have $\Pi^\uparrow_b = \cc(F_b)$, while if $b$ is a leaf of $\bb T$, then $\Pi^\downarrow_b = \cc(F_b)$. Here, $\cc(F_b)$ is the relation specifying the reachability of pairs using only those arcs in $F_b$. 

The consistency constraints are then the natural constraints that must be satisfied by any possible output $H$. We use $\lor$ to denote the transitive closure of the union of two relations. Then for parent $b$ and its two children $b', b''$, the triple $(F_b, \Pi^\downarrow_b, \Pi^\uparrow_b)$, $(F_{b'}, \Pi^\downarrow_{b'}, \Pi^\uparrow_{b'})$, $(F_{b''}, \Pi^\downarrow_{b''}, \Pi^\uparrow_{b''})$ is consistent if 
\begin{itemize}
    \item $\Pi^\downarrow_b = (\Pi^\downarrow_{b'} \lor \Pi^\downarrow_{b''} \lor \cc(F_b))[X_b]$,
    \item $\Pi^\uparrow_{b'} = (\Pi^\uparrow_b \lor \Pi^\downarrow_{b''} \lor \cc(F_{b'}))[X_{b'}]$, and 
    \item $\Pi^\uparrow_{b''} = (\Pi^\uparrow_b \lor \Pi^\downarrow_{b'} \lor \cc(F_{b''}))[X_{b''}]$.
\end{itemize}

The following claim, analogous to Claim 4.1 from \cite{DKL24}, says that these consistency constraints are enough to recover a subgraph from a set of consistent labels. 
\begin{claim}\label{claim:pullback-labeling}
If $\{(F_b, \Pi^\downarrow_b, \Pi^\uparrow_b)\}_{b \in B}$ is a consistent labeling of $\bb T$, let $H = (V, \bigcup_{b \in B} F_b)$. Then $\Pi^\downarrow_b(H) = \Pi^\downarrow_b$ and $\Pi^\uparrow_b(H) = \Pi^\uparrow_b$ for every $b \in B$.
\end{claim}

\begin{proof}
The proof is by induction. First we show $\Pi^\downarrow_b(H) = \Pi^\downarrow_b$. For each leaf bag $b$, this clearly holds since $\Pi^\downarrow_b$ is just $\cc(F_b)$. 

Then, for any non-leaf bag $b$, we have by induction that its children $b'$ and $b''$ satisfy the hypothesis. Hence, by consistency, we must have 
\begin{align}
    \Pi^\downarrow_b &= (\Pi^\downarrow_{b'} \lor \Pi^\downarrow_{b''} \lor \cc(F_b))[X_b] \nonumber \\
    &= (\Pi^\downarrow_{b'}(H) \lor \Pi^\downarrow_{b''}(H) \lor \cc(E_H \cap E_b))[X_b] \label{eq:unions}
\end{align}
Observe that $\Pi^\downarrow_b(H)$ is defined as the reachability relation within the edge set $E_H \cap \bigcup_{c \in \Lambda(b)} E_{c}$. But we can rewrite 
\[
    E_H \cap \bigcup_{c \in \Lambda(b)} E_{c} = \left(E_H \cap \bigcup_{c \in \Lambda(b')} E_{c}\right) \cup \left(E_H \cap \bigcup_{c \in \Lambda(b'')} E_c \right) \cup \left(E_H \cap E_b\right).
\]
Therefore, $\Pi^\downarrow_b(H)$ is precisely the transitive closure of the union of the reachability relations for each of the three sets in the union of the right hand side. These are $\Pi^\downarrow_{b'}(H)$, $\Pi^\downarrow_{b''}(H)$, and $\cc(E_H \cap E_b)$, appearing in \eqref{eq:unions}. So, by induction we have shown that $\Pi^\downarrow_b = \Pi^\downarrow_b(H)$ for every $b$, as desired. 

Next, we use a second induction to show that $\Pi^\uparrow_b = \Pi^\uparrow_b(H)$. This time, the base case is the root bag $b = \bb r$. Here $\Pi^\uparrow_b = \cc(F_b) = \Pi^\uparrow_b(H)$. 
The proof of the induction step is identical to before, except that we now observe that the induction hypothesis \emph{and} the above proof that $\Pi^\downarrow_b = \Pi^\downarrow_b(H)$ give, for some child $b'$ of $b$ with sibling $b''$
\begin{align*}
    \Pi^\uparrow_{b'} &= (\Pi^\uparrow_{b} \lor \Pi^\downarrow_{b''} \lor \cc(F_{b'}))[X_{b'}]\\
    &= (\Pi^\uparrow_{b}(H) \lor \Pi^\downarrow_{b''}(H) \lor \cc(E_H \cap E_{b'}))[X_{b'}]
\end{align*}
Now proceed as before. 
\end{proof}

Finally, we prove \Cref{thm:dbdst} by describing how to reduce Bounded Degree Directed Steiner Tree on bounded treewidth graphs to the Tree Label Problem. 

\begin{proof}[Proof of \Cref{thm:dbdst}]
Given an instance of Bounded Degree Directed Steiner Tree, we create a Tree Label instance on the tree decomposition $\bb T$ of $G$ (assuming the properties on $\bb T$ described above). The labels of a bag $b \in B$ are $L_b$, with consistency constraints as described.

Now we may encode the requirement that all terminals are reached from the root using the covering constraints: for each terminal $t \in S$ and each bag $b \in B$ containing $t$, the set $S_t$ contains every label $(F_b, \Pi^\downarrow_b, \Pi^\uparrow_b)$ for which $t$ is reachable from $r$ in $\Pi^\downarrow_b \lor \Pi^\uparrow_b$ (recall that $r \in X_b$ for every bag $b$).  

The degree bounds can be encoded with cost constraints: for each $v \in V$, there will be a cost $c^v$ on each label. We consider those labels $\ell = (F_b, \Pi^\downarrow_b, \Pi^\uparrow_b)$ for which $v \in X_b$ and give them a cost of $c^v_\ell = \abs{\delta_{F_b}^+(v)}$. Disallow all labels for which $c^v_\ell > d_v$, then scale all costs down by a $d_v$ factor so they lie in $[0, 1]$. Now, the $c^v$ cost constraint in the Tree Label instance corresponds exactly to the desired degree constraint. 

Next, we add one more cost constraint corresponding to the weight $W^*$ of the optimal solution (we can assume we know the weight of $\OPT$, using binary search to guess it). For a label $\ell = (F_b, \Pi^\downarrow_b, \Pi^\uparrow_b)$, we give it a cost of $c_\ell = \sum_{e \in F_b} w_e$. Disallow all labels with $c_\ell > W^*$, and scale all costs by $W^*$ so that they lie in $[0, 1]$. Thus, again, the cost constraint corresponds exactly with the weight of the output graph being at most $W^*$. 

Finally, we note that the depth of the tree is $D = O(\log n)$. And since the maximum bag size is $O(1) \cdot \tw$, then we have for any bag $b \in B$ that there are at most $\tw^{O(\tw)}$ possible forests $F_b$, and at most $2^{O(\tw^2)}$ possible choices for each of $\Pi^\downarrow_b$ and $\Pi^\uparrow_b$. Hence, in total, since treewidth is bounded, we have $\Delta = O(1)$, implying that the Tree Label algorithm from \Cref{thm:tree-label} runs in polynomial time on this instance. 

If we run the Tree Label algorithm from \Cref{thm:tree-label}, we get an output labeling of $(F_b, \Pi^\downarrow_b, \Pi^\uparrow_b)$ for each bag $b$. Let $H = (V, \bigcup_{b \in B} F_b)$. Then we know from \Cref{claim:pullback-labeling} that $\Pi^\downarrow_b$ and $\Pi^\uparrow_b$ represent the connectivity of $H$. Therefore, a covering constraint $S_t$ being satisfied by the labeling is equivalent to $H$ connecting the root $r$ to terminal $t$. So the guarantees of \Cref{thm:tree-label} say that 
\begin{itemize}
    \item $H$ connects $r$ to $t$ with probability at least $\frac{1}{D}$, for each terminal $t$.  
    \item $\E\left[\exp(\ln(1 + \frac{1}{2D}) \cdot \frac{w(H)}{W^*})\right] \leq 1 + \frac{1}{D}$. 
    \item $\E\left[\exp(\ln(1 + \frac{1}{2D}) \cdot \frac{d_H(v)}{d_v})\right] \leq 1 + \frac{1}{D}$ for every $v$. 
\end{itemize}

We may run the Tree Label algorithm from \Cref{thm:tree-label} $M = \Theta(D \log n) = \Theta(\log^2 n)$ times, and take the union $H$ of all the returned subgraphs. With high probability, $H$ connects the root to all terminals. Moreover, for each $v$, since each run of the algorithm is independent, we have $\E\left[\exp(\ln(1 + \frac{1}{2D}) \cdot \frac{d_H(v)}{d_v})\right] \leq (1 + \frac{1}{D})^M = n^{O(1)}$. So we may apply Markov's inequality and a union bound to get $\exp(\ln(1 + \frac{1}{2D}) \cdot \frac{d_H(v)}{d_v}) \leq n^{O(1)}$ for every $v$, with high probability. In particular, $d_H(v) \leq O(d_v \cdot D \ln n) = O(\log^2 n) \cdot d_v$, with high probability. Identically, we get $w(H) \leq O(\log^2n) \cdot W^*$ with high probability, completing the proof. 
\end{proof}

\section*{Acknowledgements}

D.\ Hathcock was supported by the NSF Graduate Research Fellowship grant DGE-2140739. This material is based upon work supported in part by the Air Force Office of Scientific Research under award number FA9550-23-1-0031 to R.\ Ravi. 

\bibliographystyle{alpha}
\bibliography{references2}

\end{document}